\documentclass[letterpaper, 10 pt, conference]{ieeeconf}  % Comment this line out if you need a4paper

\IEEEoverridecommandlockouts                              % This command is only needed if 
\usepackage{graphicx}          % Include this line if your 
\usepackage{epsfig} % for postscript graphics files
\usepackage{amsmath} % assumes amsmath package installed
\usepackage{amssymb}  % assumes amsmath package installed
\usepackage{hyperref}
\usepackage{mathtools}
\usepackage{cancel}
\usepackage{xcolor}
\usepackage{bm}
\usepackage{makecell}
\usepackage{tabulary,capt-of}
\usepackage{enumerate}
\usepackage{systeme}
\usepackage{algorithm}
\usepackage{algorithmic}

\usepackage{amsthm}
\usepackage{amsfonts}
\newtheorem{theorem}{Theorem}
\newtheorem{lemma}{Lemma}
\newtheorem{proposition}{Proposition}

\newtheorem{definition}{Definition}[section]

\newtheorem{remark}{Remark}[section]

\newtheorem{assumption}{Assumption}[section]

\newcommand{\fref}[1]{Fig.~\ref{#1}}

\newcommand{\aref}[1]{Assump.~\ref{#1}}
\newcommand{\lref}[1]{Lemma~\ref{#1}}
\newcommand{\tref}[1]{Theorem~\ref{#1}}
\newcommand{\cref}[1]{Corollary~\ref{#1}}
\newcommand{\dref}[1]{Definition~\ref{#1}}
\newcommand{\pref}[1]{Proposition~\ref{#1}}

\newcommand{\K}[1]{\mathbb{K}_{#1}}

\newcommand{\mat}[1]{\begin{bmatrix}\,#1\,\end{bmatrix}}

\newcommand{\norm}[2]{\lVert#1\rVert_{#2}}

\newcommand{\inner}[3]{\langle#1, #2\rangle_{#3}}

\DeclarePairedDelimiter{\Diag}{ \text{diag}\{ }{ \} }

\DeclareMathOperator*{\argmin}{argmin}

\DeclareMathOperator{\Ima}{Im}

\title{\LARGE \bf A Unified Stability Analysis of Safety-Critical Control using Multiple Control Barrier Functions}

\author{Matheus F. Reis, José P. Carvalho and A. Pedro Aguiar
\thanks{The authors are with the Research Center for Systems and Technologies (SYSTEC-ARISE), Faculdade de Engenharia, Universidade do Porto, Rua Dr. Roberto Frias, s/n 4200-465, Porto, Portugal
{\tt\small \{matheus.reis, jpcarvalho, apra\}@fe.up.pt}}
}

\begin{document}

\maketitle
\thispagestyle{empty}
\pagestyle{empty}

%%%%%%%%%%%%%%%%%%%%%%%%%%%%%%%%%%%%%%%%%%%%%%%%%%%%%%%%%%%%%%%%%%%%%%%%%%%%%%%%
\begin{abstract}
Ensuring liveness and safety of autonomous and cyber-physical systems remains a fundamental challenge, particularly when multiple safety constraints are present. 
This letter advances the theoretical foundations of safety-filter Quadratic Programs (QP) and Control Lyapunov Function (CLF)-Control Barrier Function (CBF) controllers 
by establishing a unified analytical framework for studying their stability properties. 
We derive sufficient feasibility conditions for QPs with multiple CBFs and formally 
characterize the conditions leading to undesirable equilibrium points at possible intersecting safe set boundaries.
Additionally, we introduce a stability criterion for equilibrium points, providing a systematic approach to identifying conditions under which they can be destabilized or eliminated. 
Our analysis extends prior theoretical results, deepening the understanding of the conditions of feasibility and stability of CBF-based safety filters and the CLF-CBF QP framework.
\end{abstract}

%--- Contents ------------

\section{INTRODUCTION}

Due to recent technological advancements, cyber-physical and autonomous systems, including self-driving vehicles and robotics, are becoming increasingly prevalent in our lives. 
While these systems exhibit higher autonomy levels, the increase in agency also brings more unpredictable and catastrophic failures \cite{dalrymple2024guaranteedsafeaiframework}, 
making it ever more crucial to design controllers that are both performant and provably safe.

The performance of a controller can be associated with its liveness, defined as asymptotic stability towards a set of goal states. 
On the other hand, safety can be defined as the requirement of forward invariance of the system trajectories with respect to a specified safe set. 
Combining and guaranteeing the fulfillment of these two requirements is still an open problem.
Common approaches are the safety filter architecture \cite{Hsu2024} and the CLF-CBF minimum-norm controller \cite{Ames2014}.
 In the former, the liveness and safety requirements are decoupled through a performance (nominal) controller that seeks to ensure liveness, 
 while a safety controller modifies the control law to guarantee safety at all times \cite{Wabersich2023}, making use of 
 Control Barrier Functions (CBFs) \cite {Notomista2018} to model the safety requirements.
 In the latter, safety and liveness are directly combined within an unified control framework \cite{Ames2014}, using 
 Control Lyapunov Functions (CLFs) and CBFs.
%
% Popular approaches to the safety filter include Hamilton-Jacobi (HJ) Reachability \cite{Bansal2017}
% , Predictive Safety Filters (PSF) \cite{Wabersich2021} and Control Barrier Functions (CBF) \cite{Ames2014}.
%
% Both approaches make use of Control Barrier Functions (CBFs) \cite {Notomista2018} to implement safety requirements in the form of forward invariance of system trajectories.
% popular safety filter due to their low computational complexity and  easy implementation \cite{Wabersich2023}. 

% There is a rich literature on the combination of CBFs with Control Lyapunov Functions (CLFs) through the use of quadratic programs (QPs) \cite{Ames2014}, 
% which combines safety and liveness requirements in a single control framework.

However, it is shown in \cite{Reis_LCSS} that due to the conflicting objectives of safety and liveness, the CLF-CBF QP framework 
can introduce undesirable stable equilibrium points on the safe set boundary. This means that while the framework can guarantee safety, 
it does not guarantee liveness as the system's trajectories can get stuck at undesirable equilibrium points, resulting in deadlocks. 
Research has aimed to mitigate or avoid these deadlocks in both the CLF-CBF framework and CBF-based safety filters.

In \cite{Notomista2022}, a method is proposed to solve the deadlock problem for the safety-filter considering multiple non-convex unsafe regions modeled by CBFs. 
Undesirable stable equilibria are avoided by using a nonlinear transformation that maps the system state into a new domain called the ``ball world'', 
where all non-convex CBF boundaries are converted into moving $n$-spheres and dynamically avoid the state trajectory.
In \cite{Reis_LCSS}, the authors propose a solution based on a rotating CLF, so that equilibrium points are removed in the case of a single convex unsafe set. 
The idea of modifying the CLF to avoid deadlocks is latter extended in \cite{Reis2025}, where the theory of linear matrix polynomials 
was used to compute a CLF that does not cause deadlocks, considering the class of linear systems and quadratic CLFs and CBFs.
%
% reducing the stability region of stable 
% equilibrium points to a set of zero measure.
%
% Although effective to solve deadlocks, the nonlinear transformation used by their proposed method can be highly dependent on the CBF geometry, and thus could be difficult to find for arbitrarely complex barrier shapes.
%
Finally, in \cite{kim2025learningrefineinputconstrained}, the authors address the deadlock problem by proposing a 
learning-based framework that adapts the CBF class $\mathcal{K}$ function online. Using a probabilistic ensemble neural network (PENN), 
the framework predicts safety and performance metrics, including deadlock time, which is then minimized to mitigate deadlocks.
 % {\color{red} Should I say: "Note that, unlike the other mentioned works, this approach does not derive formal guarantees for the non-existence of undesirable stable equilibrium points.}

This paper builds upon the work \cite{Reis2025} and presents a generalized theoretical understanding for the conditions for the formation 
of undesirable boundary equilibrium points with the CBF-based safety filters and CLF-CBF frameworks for safety-critical control.
Particularly:
\begin{itemize}
    \item Addresses the stability theory for the closed-loop system of affine nonlinear systems with both the safety-filter QP and CLF-CBF mininum-norm QP using a single, generalized QP framework for both safety-critical controllers.
    \item Derives sufficient conditions for the feasibility of the safety-filter QP and CLF-CBF QP with multiple CBFs.
    \item Introduces conditions for existence of undesirable equilibrium points that are valid both for the safety-filter QP and for the CLF-CBF QP, considering multiple, possibly overlapping CBF unsafe sets.  
    % {\color{red} OLD: unsafe sets modeled by CBFs 
    % where the equilibrium points can occur at boundary intersections.}
    \item Derives a novel condition for the stability of boundary equilibrium points of this type with a clear geometric intuition.
    % \item \sout{Presents numerical examples illustrating the validity of the developed theory.}
\end{itemize}
\section{PRELIMINARIES}

% -----------------------------------------------------------------------------------------------
{\it Notation}:
\label{sec:notation}
%
% The fields of real and complex numbers are $\mathbb{R}$ and $\mathbb{C}$, respectively.
%
Given a matrix $A \in \mathbb{R}^{n \times m}$ or vector $v \in \mathbb{R}^n$
% $[A]_{ij} \in \mathbb{R}$ denotes its $i$-th row, $j$-th column component and
$[A]_{k} \in \mathbb{R}^n$ denotes its $k$-th column and $[v]_{k} \in \mathbb{R}$ denotes 
its $k$-th element, while $\Diag{ \alpha_1 ,\, \cdots ,\, \alpha_r }$ is block matrix with the diagonal stacking of $r$ scalars or matrices $\alpha_1 ,\, \cdots ,\, \alpha_r$.
%
% The group of real symmetric matrices is $\mathcal{S}^n \subset \mathbb{R}^{n \times n}$.
%
% The determinant of a square matrix $A$ is $|A|$, its Frobenius norm is $\norm{A}{\mathcal{F}}$, and 
% its adjoint matrix is $\Adj{A} \in \mathbb{R}^{n \times n}$, where $A \, \Adj{A} = |A| I_n$, 
Matrix $I_n \in \mathbb{R}^{n \times n}$ is the $n \times n$ identity matrix.
%
% Given a vector $v \in \mathbb{R}^{n}$, $[v]_k \in \mathbb{R}$ is its $k$-th component.
%
% A scalar-valued function $f: \mathbb{R}^n \rightarrow \mathbb{R}$ is said to be of (differentiability) class $C^k$ if all of its $k$-th order partial derivatives exist and are continuous.
% %
% Consider the class $C^2$ function $f: \mathbb{R}^n \rightarrow \mathbb{R}$:
% %
% (i) its {\it gradient} is defined as the vector-valued function $\nabla f: \mathbb{R}^n \rightarrow \mathbb{R}^n$ such that
% $[\nabla f(x)]_k = \frac{\partial f(x)}{\partial x_k} = \partial_k f(x)$, where $\partial_k$ denotes partial differentiation with respect to the $k$-th component of the function input,
% (ii) its {\it Hessian} matrix is defined as the matrix-valued function $H_f: \mathbb{R}^n \rightarrow \mathcal{S}^n$ such that $[H_f(x)]_{ij} = \frac{\partial^2 f(x)}{\partial x_i \partial x_j}$.
%
$L_g f$ is the Lie derivative of a function $f$ along a function $g: \mathbb{R}^{n} \rightarrow \mathbb{R}^{n \times m}$, that is, $L_g f = \nabla f^\mathsf{T} g \in \mathbb{R}^{m}$.
%
%A level set of function $f$ is of the form $L_c(f) = \{ x \in \mathbb{R}^{n} : f(x) = c \}$, that is, the set where $f$ takes on a given constant $c$.
%
%\vspace{-3mm}
The inner product between two vectors $u, v \in \mathbb{R}^n$ induced by a positive semidefinite matrix $X = X^\mathsf{T} \ge 0$ is given by $\inner{u}{v}{X} = u^\mathsf{T} X \, v$.
This inner product induces a norm $\norm{v}{X}^2 = \inner{v}{v}{X} = v^\mathsf{T} X \, v$ over $\mathbb{R}^n$. 
The standard inner product is then $\inner{u}{v}{} = \inner{u}{v}{I_n}$, with standard Euclidean norm $\norm{v}{}^2 = \norm{v}{I_n}^2$.
The orthogonal complement of a subspace $\mathcal{W}$ is denoted by $\mathcal{W}^\perp$, with orthogonality dependent on an inner product $\inner{\cdot}{\cdot}{X}$.
%
% The set 
% $\Span{v_1, \cdots, v_p}$ is the set of all linear combinations of vectors from $\{v_1, \cdots, v_p \} \subset \mathbb{R}^n$. 
%
% The positive semi-definite cone of symmetric matrices is $\mathcal{S}^n_+$.
%
%The set of semidefinite symmetric matrices is given by $\mathcal{S}^n_+ \cup \mathcal{S}^n_-$.
%
% The null space and spectrum of a real square matrix $A \in \mathbb{R}^{n \times n}$ are given by $\mathcal{N}(A) \subset \mathbb{R}^n$ and $\sigma(A) \subset \mathbb{R}$, respectively.

%--------------------------------------------------------------------------------------------------
\subsection{Control Lyapunov (CLFs) and Barrier Functions (CBFs)}

Consider the nonlinear control affine system
\begin{align}
\dot{x} = f(x) + g(x) u
\label{eq:affine_nonlinear}
\end{align}
where $x \in \mathbb{R}^n$, $u \in \mathbb{R}^m$ are the system state and control input, respectively, and 
$f: \mathbb{R}^n \rightarrow \mathbb{R}^n$, $g: \mathbb{R}^n \rightarrow \mathbb{R}^{n \times m}$ are locally Lipschitz.
%
%The notation of $G(x) = g(x) g(x)^\mathsf{T}$ will also prove to be useful.
%
%Next, we recall the definitions of control Lyapunov Functions (CLFs) \cite{Sontag1983} and control barrier functions (CBFs) \cite{Ames2019}.

%System \eqref{eq:nonlinear_system} is said to be stabilizable to the origin $x = \bm{0}_n$ if there exists a control input $u(t)$ such that $\dot{V} \le -\gamma(V(x))$ along the system trajectories, where $\gamma: \mathbb{R}_{\ge 0} \rightarrow \mathbb{R}_{\ge 0}$ represents a class $\mathcal{K}$ function.
%
%Here, $L_f V = f^\mathsf{T} \nabla V$, $L_g V = g^\mathsf{T} \nabla V$ are {\it Lie derivatives} of $V$ with respect to $f$ and $g$, respectively.

%suppose we have the control objective of asymptotically stabilizing the nonlinear control system \eqref{eq:nonlinear_system} to the origin $x^* = 0$. In the CLF context, this can be understood by finding a feedback control law that drives a positive definite function, $V: D \subset \mathbb{R}^n \rightarrow \mathbb{R}_{\ge 0}$ to zero.
%%
%That is, if
%%
%\begin{align}
%\exists\,u  \,\, s.t. \,\, \dot{V} = L_f V(x) + L_g V(x)^\mathsf{T} u \le -\gamma(V(x))
%\label{eq:CLF_cond}
%\end{align}
%%
%the system is {\it stabilizable} to $V(x^*) = 0$, i.e., $x^* = 0$, where $\gamma: \mathbb{R}_{\ge 0} \rightarrow \mathbb{R}_{\ge 0}$ is a class $\mathcal{K}$ function. Here, $L_f V = f^\mathsf{T} \nabla V$, $L_g V = g^\mathsf{T} \nabla V$ are {\it Lie derivatives} of $V$ with respect to $f$ and $g$, respectively.
\begin{definition}[CLFs]
\label{def:CLF}
A positive definite function $V$ is a {\it control Lyapunov function} (CLF) for system \eqref{eq:affine_nonlinear} if it satisfies:
%
% \vspace{-4mm}
\begin{align}
\inf_{u\in \mathbb{R}^m} \left[L_f V(x) + L_g V u \right] \le -\gamma(V(x)) \nonumber
%\label{eq:CLF_cond2}
\end{align}
where $\gamma: \mathbb{R}_{\ge 0} \rightarrow \mathbb{R}_{\ge 0}$ is a class $\mathcal{K}$ function \cite{Khalil2002}.
\end{definition}
\dref{def:CLF} implies that there exists a set of stabilizing controls that makes the CLF strictly decreasing everywhere outside its global minimum $x_0 \in \mathbb{R}^n$, given by 
$\K{V}(x) = \{ u \!\in\! \mathbb{R}^m: L_f V + L_g V u \le - \gamma(V) \} $.
%
% \begin{align}
% \K{V}(x) = \{ u \!\in\! \mathbb{R}^m: L_f V + \inner{L_g V}{u}{} \le - \gamma(V) \} 
% \label{eq:CLF_set}
% \end{align}
%
% where $F_V(x) = L_f V(x) + \gamma(V(x))$. 
% Note that this is an affine constraint on $u \in U$, defining a half-plane in the control space $u$. 
% This fact will allow for the formulation of optimization-based controllers.

% In some sense, the concept of {\it safety} can be thought as the dual of {\it stabilization}: while the purpose of the latter is to guide the system towards a goal state, 
% the purpose of the former is to ensure that the system trajectories will never visit undesired states.

%--------------------------------------------------------------------------------------------------
\begin{definition}[Safety]
\label{safe}
The trajectories of a given system are safe with respect to a set $\mathcal{C}$ if $\mathcal{C}$ is forward invariant, meaning that for every $x(0) \in \mathcal{C}$, $x(t) \in \mathcal{C}$ for all $t>0$.
\end{definition} 
%--------------------------------------------------------------------------------------------------
%
Consider $N$ subsets $\mathcal{C}_1, \dots, \mathcal{C}_N \subset \mathbb{R}^n$  defined by the superlevel set of $N$ 
continuously differentiable functions $h_i : \mathbb{R}^n \rightarrow \mathbb{R}$, $i =1,2,\ldots, N$
% , namely $\mathcal{C}_i = \{x \in \mathbb{R}^n: h_i(x) \ge 0\}$.
%
% \begin{align}
% \mathcal{C}_i &= \{x \in \mathbb{R}^n: h_i(x) \ge 0\}
% \label{eq:safe_set}
% % \partial \mathcal{C}_i &= \{x \in \mathbb{R}^n: h_i(x) = 0\} \label{eq:boundary}
% %%\text{int}(\mathcal{C}_i) &= \{x \in \mathbb{R}^n: h_i(x) > 0\} \,. \nonumber
% \end{align}
%
The corresponding $i$-th boundary is given by $\partial \mathcal{C}_i = \{x \in \mathbb{R}^n: h_i(x) = 0\}$.

%--------------------------------------------------------------------------------------------------
\begin{definition}[CBFs]
Let $\mathcal{C}_i = \{x \in \mathbb{R}^n: h_i(x) \ge 0\}$. Then $h_i(x)$ is a (zeroing) {\it Control Barrier Function} 
(CBF) for \eqref{eq:affine_nonlinear} if there exists a locally Lipschitz extended class $\mathcal{K}_{\infty}$ function \cite{Khalil2002}
% \footnote{An extended class $\mathcal{K}_{\infty}$ function $\alpha_i: \mathbb{R} \rightarrow \mathbb{R}$ is strictly increasing with
% $\alpha_i(0) = 0$.}
$\alpha_i : \mathbb{R} \rightarrow \mathbb{R}$ such that
\begin{align}
\sup_{u\in \mathbb{R}^m} \left[L_f h_i(x) + \inner{L_g h_i(x)}{u}{} \right] \ge -\alpha_i(h_i(x)) 
\nonumber
%\label{eq:CBF_cond}
\end{align}
\end{definition}
%--------------------------------------------------------------------------------------------------
This definition means that there exists a set of safe controls allowing the $i$-th CBF to decrease on in the interior of its safe set $\text{int}(\mathcal{C}_i)$,
but not on its boundary $\partial \mathcal{C}_i$, given by 
$\K{h_i}(x) \!=\! \{ u \in \mathbb{R}^m\!:\! L_f h_i \!+\! L_g h_i u \!+\! \alpha_i(h_i) \ge 0 \}$.
%
% \begin{align}
% \K{h_i}(x) \!&=\! \{ u \in \mathbb{R}^m\!:\! L_f h_i \!+\! \inner{L_g h_i}{u}{} \!+\! \alpha_i(h_i) \ge 0 \}
% \label{eq:CBF_set}
% \end{align}
%
% where $F_{h_i} = L_f h_i(x) + \alpha_i(h_i(x))$.
%
The composite safe set associated to all the $N$ CBFs is simply $\mathcal{C} = \bigcap^N_{i=1} \mathcal{C}_i$
and the set of controls rendering $\mathcal{C}$ forward invariant is the intersection 
$\K{h}(x) = \bigcap^N_{i=1} \K{h_i}(x)$.
%
% \begin{align}
% \K{h}(x) &= \bigcap^N_{i=1} \K{h_i}(x)
% \label{eq:multiple_CBF_set}
% \end{align}

%--------------------------------------------------------------------------------------------------
\subsection{QP-Based Safety-Critical Controllers}

Consider the closed-loop system for \eqref{eq:affine_nonlinear} 
\begin{align}
\dot{x} = f_{cl}(x) := f(x) + g(x) u^\star(x)
\label{eq:closed_loop}
\end{align}
with a state-feedback control law $u^\star(x)$.
%
% As pointed out by \cite{Ames2014}, the fact that the CLF and CBF constraints
% are affine in the control $u$ allows for a formulation of safety-critical controllers using QPs.
%
A common approach for safety-critical control is the safety-filter algorithm, which minimally 
modifies a given stabilizing state-feedback control law $u_{nom}(x)$ to achieve safety.
%
% \begin{align}
% u^\star(x) = \argmin_{u\in\mathbb{R}^{m}}&\frac{1}{2} \norm{u - u_{nom}(x)}{}^2 \qquad 
% \label{eq:safety_filter_QP} \\
% s.t. \, L_f h_i &+ \inner{L_g h_i}{u}{} \ge -\alpha_i(h_i) \tag{CBFs}
% % \\
% % &\,\,\, \vdots \nonumber
% % \\
% % F_{h_i}(x) &+ \inner{L_g h_i}{u}{} \ge 0 \tag{N-th CBF} \nonumber
% \end{align}
%
% with $N$ CBF constraints, that is, $i \in \{ 1, \cdots, N \}$. 
% If feasible, 
The safety-filter effectively generates the ``closest'' safe control $u^{\star}(x) \in \K{h}(x)$ 
to the stabilizing control $u_{nom}(x)$.

An alternative approach that also makes use of QPs is the minimum-norm CLF-CBF QP by \cite{Ames2014}:
\begin{align}
u^\star(x), \delta^\star(x) &= \argmin_{(u,\delta)} \frac{1}{2} (\Delta u)^\mathsf{T} H(x) (\Delta u) + \frac{1}{2} p \delta^2 
\label{eq:generalized_QP} \\
s.t. \, 
&L_f V + \inner{L_g V(x)}{u}{} \le -\gamma(V) + \delta \tag{CLF} \\
&L_f h_i + \inner{L_g h_i(x)}{u}{} \ge -\alpha_i(h_i) \tag{CBFs}
% \\
% &\,\,\, \vdots \nonumber
% \\
% F_{h_i} + \inner{L_g h_i}{u}{} \ge 0 \tag{N-th CBF} \nonumber
\end{align}
$i \in \{ 1, \cdots, N \}$, $H: \mathbb{R}^n \rightarrow \mathbb{R}^{n \times n}$ 
being a symmetric positive definite matrix function of the state, $p > 0$ and (for now) $\Delta u = u$.
% $\gamma$ and $\alpha$ being class $\mathcal{K}$ and class $\mathcal{K}_\infty$ functions, respectively.
%
The relaxation variable $\delta$ in the CLF constraint {\it softens} the stabilization objective, aiming to maintain the feasibility of the QP.
%
% The objective of the optimization problem in \eqref{eq:generalized_QP} is to minimize the norm of the control signal and of an auxiliary {\it relaxation} variable $\delta$ while satisfying the CLF  and CBF constraints.
%
If feasible, the feedback controller \eqref{eq:generalized_QP} with $\Delta u = u$ generates a minimum-norm stabilizing and safe control, 
guaranteeing local stability of the origin and safety of the closed-loop system trajectories with respect to the composite safe set $\mathcal{C}$.

As pointed out by the works \cite{Reis_LCSS,Reis2025,TanDimos2024},  neither the safety-filter
nor the CLF-CBF QP \eqref{eq:generalized_QP} (with $\Delta u = u$) can guarantee global stabilization of trajectories 
towards the origin for the closed-loop system \eqref{eq:closed_loop}, meaning that trajectories could converge towards 
undesirable equilibrium points.
% and even other types of attractors, such as limits cycles.
%
Here, our objective is to study the solutions, existence conditions and stability of equilibrium points of the closed-loop system formed by 
a generalized QP controller given by \eqref{eq:generalized_QP} with $\Delta u = u - u_{nom}(x)$, where $u_{nom}$ is the stabilizing control 
law from the safety-filter QP.
%
% \begin{align}
% u^\star(x), \delta^\star(x) &= \argmin_{(u,\delta)} \frac{1}{2} (\Delta u)^\mathsf{T} H(x) (\Delta u) + \frac{1}{2} p \delta^2 
% \label{eq:generalized_QP} \\
% s.t. \, 
% &L_f V + \inner{L_g V(x)}{u}{} \le -\gamma(V) + \delta \tag{CLF} \\
% &L_f h_i + \inner{L_g h_i(x)}{u}{} \ge -\alpha_i(h_i) \tag{CBFs}
% % \\
% % &\,\,\, \vdots \nonumber
% % \\
% % F_{h_i} + \inner{L_g h_i}{u}{} \ge 0 \tag{N-th CBF} \nonumber
% \end{align}
% %
% with the $N$ CBF constraints, that is, $i \in \{ 1, \cdots, N \}$, and 
%
The structure of controller \eqref{eq:generalized_QP} with $\Delta u = u - u_{nom}(x)$ 
allows for the generalization of the safety-filter and the CLF-CBF QPs into a single framework:
\begin{enumerate}
    \item with $H(x) = I_m$ and $V = 0$ (without CLF), we recover the usual safety-filter QP. 
    In this case, the optimal solution for the slack variable is always $\delta^\star = 0$.
    \item with $u_{nom}(x) = 0$, we recover the usual CLF-CBF QP.
\end{enumerate}

If feasible, the solution of the generalized QP \eqref{eq:generalized_QP} is guaranteed $u^\star(x) \in \K{h}(x)$, and possibly 
close to the stabilizing set $\K{V}(x)$. However, due to the slack variable, it is not possible to strictly guarantee that
$u^\star(x) \in \K{V}(x) \cap \K{h}(x) \,\, \forall x \in \mathbb{R}^n$. 
That means that safety with respect to $\mathcal{C}$ is achieved, but stabilization is (possibly) hampered.

\begin{assumption}
\label{assumption:initial_state}
The initial state $x(0) \in \mathbb{R}^n$ and the origin $0 \in \mathbb{R}^n$ are contained in the safe set $\mathcal{C}$, that is, 
$h_i(x(0)) \ge 0$ and $h_i(0) \ge 0$ for all $i \in \{ 1, \cdots, N \}$.
\end{assumption}
%
% Assumption \ref{assumption:initial_state} comes from the fact that it is natural to assume that a system starts in a safe configuration; as an example, 
% it is only natural to assume that a vehicle starts its navigation task in a safe state of non-collision against obstacles.
% Furthermore, since the origin is the goal state, it must be reachable by the closed-loop system. 
%--------------------------------------------------------------------------------------------------
\begin{theorem}
\label{thm:QP_feasibility}
Under \aref{assumption:initial_state}, consider the following assumptions:\\
{\bf(i)} There is only one CBF ($N=1$).\\
{\bf(ii)} System \eqref{eq:affine_nonlinear} is driftless: $f(x) = 0 \,\,\, \forall x \in \mathbb{R}^n$.\\
{\bf(iii)} Considering any number of CBFs ($N>1$),
\begin{align}
    b_2(x) &\in \Ima{U^\mathsf{T} \mathcal{P}_V G U} \label{eq:feasibility_cond} \\
    b_2(x) &= U^\mathsf{T} \left( c^{-1} \gamma(V) G \nabla V - \mathcal{P}_V f_{nom} \right) - \bar{\alpha} \nonumber
\end{align}
where $c \!=\! p^{-1} \!+\! \norm{\nabla V}{G}^2 > 0$, 
$G(x) \!=\! g(x) H(x)^{-1} g(x)^\mathsf{T}$, $f_{nom}(x) \!=\! f(x) \!+\! g(x) u_{nom}(x)$, 
$\mathcal{P}_V = I - c^{-1} G \nabla V \nabla V^\mathsf{T}$,
$\bar{\alpha} \!=\! \mat{\alpha_1(h_1) \!\!\!&\!\!\! \cdots \!\!\!&\!\!\! \alpha_N(h_N)}^\mathsf{T}$ and
$U(x) \!=\! \mat{\nabla h_1 \!\!\!&\!\!\! \cdots \!\!\!&\!\!\! \nabla h_N}$.\\
Then, QP \eqref{eq:generalized_QP} is feasible under Assumptions {\bf (i)} {\it or} {\bf (ii)} {\it or} {\bf (iii)}.
% Any of the conditions above is sufficient for the feasibility of 
\end{theorem}
%--------------------------------------------------------------------------------------------------
\begin{proof}
% This demonstration is divided in the following cases.\\
The proofs of {\bf(i)} and {\bf(ii)} can be found in \cite{Ames2014} and \cite{Reis2025}, respectively.
To establish the proof of {\bf(iii)}, we begin by formulating the Lagrangian associated with the QP \eqref{eq:generalized_QP} 
\begin{align}
\mathcal{L}(u, \delta, \bar{\lambda}) \!&=\! \frac{1}{2} \left( \norm{\Delta u}{H}^2 \!+\! p \delta^2 \right) 
\label{eq:general_lagrangian} \\
&\!+ \lambda_0 ( L_f V \!+\! \inner{L_g V}{u}{} \!+\! \gamma(V) \!-\! \delta ) \nonumber \\
&\!- \sum^N_{i=1} \lambda_i ( L_f h_i \!+\! \inner{L_g h_i}{u}{} \!+\! \alpha_i(h_i) ) 
\nonumber
\end{align}
where $\lambda_i \ge 0$ and 
$\lambda \!=\! \mat{ \lambda_1 \!\!\!&\!\!\! \cdots \!\!\!&\!\!\! \lambda_N }^\mathsf{T} \in \mathbb{R}^N_{\ge 0}$,
$\bar{\lambda} \!=\! \mat{ \lambda_0 \!\!&\!\! \lambda^\mathsf{T} }^\mathsf{T} \in \mathbb{R}^{N+1}_{\ge 0}$
are vectors of KKT multipliers associated to the optimization problem.
Using matrix $U$, the stationarity KKT conditions give the following solutions for the QP:
%
% \begin{align}
% \frac{\partial \mathcal{L}}{\partial u} 
% = H \Delta u + g^\mathsf{T} \left( \lambda_0 \nabla V - U \lambda \right) &= 0 \label{eq:KKT1} \\
% \frac{\partial \mathcal{L}}{\partial \delta} = p \delta - \lambda_0 &= 0 \label{eq:KKT2}
% % \lambda_0 ( F_{V} + L_g V \,u - \delta ) &= 0 \label{eq:KKT3} \\
% % \lambda_i ( F_{h_i} + L_g h_i \,u ) &= 0 \label{eq:KKT4}
% \end{align}
%
% Using \eqref{eq:KKT1}-\eqref{eq:KKT2}, the QP solutions must satisfy:
%
\begin{align}
u^\star(x) &= u_{nom} + H^{-1} g^\mathsf{T} \left( -\lambda_0 \nabla V + U \lambda \right) \label{eq:control_solution} \\
\delta^\star(x) &= p^{-1} \lambda_0 \label{eq:delta_solution}
\end{align}
The dual function $g(\bar{\lambda}) = \min_{(u,\delta)} \mathcal{L}(u, \delta, \bar{\lambda})$ associated to the QP can be obtained by 
substituting \eqref{eq:control_solution} into \eqref{eq:general_lagrangian}, yielding the following dual QP:
\begin{align}
\max_{\bar{\lambda} \in \mathbb{R}^N_{\ge 0} }& -\frac{1}{2} 
\bar{\lambda}^\mathsf{T} A(x) \bar{\lambda}
+ \bar{\lambda}^\mathsf{T} b(x)
\label{eq:dualQP} \\
% s.t. \,& \lambda_0, \cdots, \lambda_N \ge 0 \nonumber \\
A(x) &= 
\begin{bmatrix}
c \!\!&\!\! - \nabla V^\mathsf{T} G U \\
- U^\mathsf{T} G \nabla V \!\!&\!\! U^\mathsf{T} G U
\end{bmatrix} \,, \,\, 
b(x) = 
\begin{bmatrix}
F_{V} \\ 
-F_{h}
\end{bmatrix} \nonumber
\end{align}
with $F_{V} = L_{f_{nom}} V + \gamma(V)$, $F_{h} = U^\mathsf{T} f_{nom} + \bar{\alpha}$.
Notice that the dual cost in \eqref{eq:dualQP} is bounded from above if and only if
$
b(x) \in \Ima{A(x)}
$. 
In that case, the primal QP \eqref{eq:generalized_QP} is feasible. Applying Gauss elimination to the 
augmented matrix $\mat{ A(x) \!\!&\!\!\!\!|\!\!\!\!&\!\! b(x) }$ yields
\begin{align}
    \begin{bmatrix}
        1 \!\!&\!\! - c^{-1} \nabla V^\mathsf{T} G U \!\!\!&|&\!\!\! c^{-1} F_V \\
        0 \!\!&\!\!  U^\mathsf{T} \mathcal{P}_V G U \!\!\!&|&\!\!\! b_2(x)
    \end{bmatrix}
    \label{eq:augmented_system}
\end{align}
%
% where the first line of $\mat{ A(x) \!\!&\!\!\!\!|\!\!\!\!&\!\! b(x) }$ was first divided by $c > 0$ and then
% summed to the remaining lines appropriately multiplied by the elements of the first column
% in order to obtain a column of zeros below the pivot element.
%
Then, from \eqref{eq:augmented_system}, the condition $b_2(x) \in \Ima{U^\mathsf{T} \mathcal{P}_V G U}$ 
as stated in \eqref{eq:feasibility_cond} is equivalent to $b(x) \in \Ima{A(x)}$, meaning that the primal 
QP is feasible under this condition.
\end{proof}
%--------------------------------------------------------------------------------------------------

\tref{thm:QP_feasibility}{\bf(iii)} provides a sufficient condition for the feasibility of QP \eqref{eq:generalized_QP}: 
when \eqref{eq:feasibility_cond} holds, the corresponding dual QP cost is bounded, which means \eqref{eq:generalized_QP} is feasible.  
Feasibility condition \eqref{eq:feasibility_cond} is of particular importance when assumptions {\bf(i)} and {\bf(ii)} fail,
giving a sufficient condition for the feasibility of QP \eqref{eq:generalized_QP} in the case when multiple safety 
objectives are simultaneously required.
% In the next sections, we assume that the generalized QP \eqref{eq:generalized_QP} is feasible for all system trajectories.
\section{Closed-loop System Analysis}

\subsection{Existence of Equilibrium Points}

In this section, we extend a result from \cite{Reis_LCSS}, regarding the existence of equilibrium points in the 
safety-filter QP or in the CLF-CBF QP when multiple CBF constraints are present. The results are conditioned to 
the feasibility of the QP \eqref{eq:generalized_QP}: that is, it is assumed that at least one of the conditions 
of \tref{thm:QP_feasibility} holds.
%
% Particularly, Theorem \ref{theorem1} shows that the driftless, affine closed-loop system with the QP-based controller can have undesirable equilibrium points other than the CLF global minimum, and all of them are located on the boundaries of the safe sets.
% %
% In this work, for generality purposes, we assume the CLF global minimum to be located at an arbitrary point $x_0 \in \mathbb{R}^n$ such that $V(x_0) = 0$.

%--------------------------------------------------------------------------------------------------
\begin{definition}[Equilibrium Manifold]
\label{def:equilibrium_manifold}
Let the set $\mathcal{A} = \{a_1, \cdots, a_r\} \subset 2^{\{1, \cdots, N\}}$ be a collection of $1 \le r \le N$ CBF 
indexes corresponding to $r$ CBFs $\{ h_{a_1}, \cdots, h_{a_r} \}$.
Define the vector field $f_{\mathcal{A}}: \mathbb{R}^n \times \mathbb{R}^r_{\ge 0} \rightarrow \mathbb{R}^{n}$:
\begin{align}
f_{\mathcal{A}}(x, \lambda) &= f_{nom} - p \gamma(V) G \nabla V + G U_{\!\mathcal{A}} \lambda \label{eq:fa} \\
U_{\!\mathcal{A}} &= \mat{ \nabla h_{a_1} \!\!&\!\! \cdots \!\!&\!\! \nabla h_{a_r} } \in \mathbb{R}^{n \times r} \label{eq:Umatrix}
% U_{\lambda_a}(x) &= \mat{\nabla h_{a_1} \!\!\!&\!\!\! \cdots \!\!\!&\!\!\! \nabla h_{a_r}} \in \mathbb{R}^{n \times r}
\end{align}
%
% The Jacobian matrix of $f_a$ with respect to $x \in \mathbb{R}^n$ is
% %
% \begin{align}
% J_{f_a}(x, \lambda_a) =
% &= \frac{\partial f}{\partial x}
% + \lambda \frac{\partial (G U_{\lambda_a})}{\partial x}
% - p \frac{\partial (\gamma(V) G \nabla V)}{\partial x}
% \label{eq:Jfa}
% \end{align}
%
% where $\lambda_a = \mat{\lambda_{a_1} \!\!&\!\! \cdots \!\!&\!\! \lambda_{a_r}}^\mathsf{T} \in \mathbb{R}^{r}_{\ge 0}$.
\end{definition}
%-------------------------------------------------------------------------------------------------
As will be demonstrated in the next sections, \eqref{eq:fa} and its Jacobian with respect to $x$ 
will be of central importance to characterize the existence and stability conditions for the 
equilibrium points of the closed-loop system.
%------------------------------------------------------------------------------------------------
\begin{theorem}[Existence of Equilibrium Points]
\label{theorem:existence_equilibria}
Let \eqref{eq:closed_loop} be the closed-loop system formed by the nonlinear system \eqref{eq:affine_nonlinear} with controller \eqref{eq:generalized_QP}.
Let the set $\mathcal{A}$ as defined in \dref{def:equilibrium_manifold} representing the indexes of 
$r$ overlapping CBF boundaries $\partial \mathcal{C}_{\mathcal{A}} = \bigcap^r_{i=1} \partial \mathcal{C}_{a_i} \ne \emptyset$.
%
% \begin{align}
%     \partial \mathcal{C}_{\mathcal{A}} = \bigcap^r_{i=1} \partial \mathcal{C}_{a_i} \ne \emptyset
%     \label{eq:boundary_intersection}
% \end{align}
%
The equilibrium points of \eqref{eq:closed_loop} come in two distinct types:
%
% The set $\mathcal{E}$ of equilibrium points of \eqref{eq:closed_loop} is
%
\begin{align}
% \mathcal{E} \!&=\! \left( \bigcup_{\mathcal{A}} \mathcal{E}_{\partial \mathcal{C}_{\mathcal{A}}} \right) \cup \mathcal{E}_{int(\mathcal{C})}
% \label{eq:equilibrium_points_set} \\
\mathcal{E}_{\partial \mathcal{C}_{\mathcal{A}}} \!&=\! \partial \mathcal{C}_{\mathcal{A}} \!\cap\! \{ x \in \mathbb{R}^n \,|\, \exists \lambda \!\in\! \mathbb{R}^r_{\ge 0} \text{ s.t. } f_{\mathcal{A}}(x, \lambda) \!=\! 0 \} \label{eq:boundary_equilibria} \\
\mathcal{E}_{int(\mathcal{C})} \!&=\! int(\mathcal{C}) \!\cap\! \{ x \in \mathbb{R}^n \,|\, f_{nom} \!=\! p \gamma(V) G \nabla V \}
\label{eq:interior_equilibria}
\end{align}
where 
$\mathcal{E}_{\partial \mathcal{C}_{\mathcal{A}}}$ is the set of {\it boundary} equilibrium points occurring in $\partial \mathcal{C}_{\mathcal{A}}$ 
and 
$\mathcal{E}_{int(\mathcal{C})}$ is the set of {\it interior} equilibrium points.
Furthermore, defining the region of the state space where only the $a_1, \cdots, a_r$ CBF constraints are simultaneously active as 
$\mathcal{S}_{\mathcal{A}} \subset \mathbb{R}^n$, we have $\mathcal{E}_{\partial \mathcal{C}_{\mathcal{A}}} \subset \mathcal{S}_{\mathcal{A}}$.

%\begin{align}
%\mathcal{E}_{\partial \mathcal{C}_i} &= \partial \mathcal{C}_i \cap \A{i} \,, \A{i} = \{ x \in \mathbb{R}^n \,|\, \exists \lambda \ge 0 \text{ s.t. } f_i(x, \lambda) = 0 \}
%\label{eq:boundary_equilibria} \\
%%\A{i} &= \big\{ x \in \mathbb{R}^n \,|\, \exists \lambda \ge 0 \text{ s.t. } 
%%f_i(x, \lambda) = 0 \big\} 
%%\label{eq:codirectionality} \\
%\mathcal{E}_{int(\mathcal{C})} &= int(\mathcal{C}) \cap \A{0} 
%\label{eq:interior_equilibria}
%%\A{0} &= \big\{ x \in \mathbb{R}^n \,|\, 
%%f(x) = p G(x) \nabla \overline{V}(x) \big\}
%%\label{eq:lambda0_codirectionality}
%\end{align}
%
%$f_i$ being the vector field defined in \eqref{eq:fi}.
%
\end{theorem}
%-------------------------------------------------------------------------------------------------
\begin{proof}
Substituting \eqref{eq:control_solution} in \eqref{eq:closed_loop} and using the definition of $G(x)$
yields the following closed-loop dynamics:
\begin{align}
f_{cl}(x) = f_{nom}(x) + G(x) \left( - \lambda_0 \nabla V + U \lambda \right)
\label{eq:closed_loop_dynamics}
\end{align}
At an equilibrium point $x_e \in \mathcal{E}$, $f_{cl}(x_e) = 0$. Applying this condition to \eqref{eq:closed_loop} yields
\begin{align}
f_{nom}(x_e) = G(x_e) \left( \lambda_0 \nabla V(x_e) - U(x_e) \lambda \right)
\label{eq:equilibrium_condition}
\end{align}
%
% Regarding the complementary slackness conditions \eqref{eq:KKT3}-\eqref{eq:KKT4}, different cases must be considered, depending on the activation of the constraints.
% %
% First, consider the two cases regarding the activation of the CLF constraint:
%---------------------------------------------------------------------------------------------------------------
{\bf Case 1.} Consider the region of the state space where the CLF constraint is {\it inactive}: 
$L_{f_{cl}} V + \gamma(V) - \delta^\star < 0$. Following the exact same steps of {\bf Case 1} 
of Theorem 2 on \cite{Reis2025}, we conclude that no equilibrium points can occur at this region.
% From the complementary slackness condition, $\lambda_0 = 0$. 
% Then, using \eqref{eq:delta_solution}, notice that $\delta^\star(x) = 0$. 
% At an equilibrium point $x_e \in \mathcal{E}$, $L_{f_{cl}} V(x_e) = 0$, 
% and therefore we obtain $\gamma(V(x_e)) < \delta^\star(x_e) = 0$, implying that 
% $V(x_e) < 0$, which is a contradiction since $V$ is a nonnegative function. 
% Therefore, all equilibrium points must lie on the region where the CLF constraint is {\it active}.\\
%------------------------------------------------------------------------------------------------
{\bf Case 2.} Consider the region where CLF constraint is {\it active}: $L_{f_{cl}} V + \gamma(V) = \delta^\star$.
In the case of the safety-filter, since $V = 0$, the CLF constraint becomes simply $0 \le \delta$ and since the 
optimal solution for the relaxation variable is always $\delta^\star = 0$, the CLF constraint is always active.
%
% For an equilibrium point $x_e \in \mathcal{E}$ occurring at this region, 
% following similar steps of {\bf Case 2.} of Theorem 2 on \cite{Reis2025}, we conclude that
%
At an equilibrium point $x_e \in \mathcal{E}$, $L_{f_{cl}} V(x_e) = 0$. Therefore, using \eqref{eq:delta_solution}, 
$\gamma(V(x_e)) = \delta^\star(x_e) = p^{-1} \lambda_0$. Then, at any equilibrium point $x_e \in \mathcal{E}$, the 
multiplier associated to the CLF constraint is $\lambda_0(x_e) = p \gamma(V(x_e)) \ge 0$.
Therefore, equation \eqref{eq:equilibrium_condition} yields
\begin{align}
f_{nom}(x_e) \!=\! G(x_e) \left( p \gamma(V(x_e)) \nabla V(x_e) - U(x_e) \lambda \right)
\label{eq:explicit_equilibrium_condition}
\end{align}
%
% Then, \eqref{eq:explicit_equilibrium_condition} shows that $\mathcal{E} \subset \A{}$, where the scalars $\lambda \ge 0$ and $p \gamma V(x_e)$ are the KKT multipliers of the optimization problem \eqref{eq:QP_control} at $x_e \in \mathcal{E}$.
For the safety-filter, $V = 0$ and $\nabla V = 0$ are valid in \eqref{eq:explicit_equilibrium_condition}.
For the next cases, the CLF constraint is assumed to be active.\\
%---------------------------------------------------------------------------------------------------------------
{\bf Case 3.} Consider the region where exactly $r \le N$ CBF constraints are simultaneously active.
Their corresponding indexes are denoted by the set $\mathcal{A} = \{a_1, \cdots, a_{r}\}$ as defined previously.
Therefore, $L_{f_{cl}} h_{i} + \alpha_i(h_{i}) = 0$, for $i \in \mathcal{A}$.
At an equilibrium point $x_e$ occurring in this region, $L_{f_{cl}} h_i(x_e) = 0$, implying that $h_{a_1}(x_e) = \cdots = h_{a_r}(x_e) = 0$.
Therefore, $x_e$ must lie at the boundary intersection associated to the $r$ active CBFs, that is, 
$x_e \in \partial \mathcal{C}_{\mathcal{A}}$.
%
% In particular, that means that if $r = |\mathcal{A}| = 1$ with $\mathcal{A} = \{a_1\}$ (only the $a_1$-th CBF is active) 
% and $\partial \mathcal{C}_{a_i}$ is disjoint to all other boundaries, any equilibrium point $x_e \in \partial \mathcal{C}_{a_i}$ 
% must necessarily occur on the region where {\it only} the $a_1$-th CBF constraint is active.
%
The conclusion is that boundary equilibrium points at $\partial \mathcal{C}_{\mathcal{A}}$ can only occur at 
$\mathcal{S}_{\mathcal{A}}$, that is, the region where {\it only} the CBF constraints corresponding to 
$h_{a_1}, \cdots, h_{a_r}$ are simultaneously active.
Since in this case the remaining CBF constraints are all inactive, $\lambda_i = 0 \,\, \forall i \notin \mathcal{A}$, 
and \eqref{eq:explicit_equilibrium_condition} reduces to 
\begin{align}
f_{nom}(x_e) = G(x_e) \left( p \gamma(V) \nabla V(x_e) - U_{\mathcal{A}}(x_e) \lambda_a \right)
\label{eq:practical_equilibrium_condition}
\end{align}
where $U_{\mathcal{A}}$ is given by \eqref{eq:Umatrix} and $\lambda_a \in \mathbb{R}^r_{\ge 0}$ is a vector of appropriate 
size with the non-negative corresponding KKT multipliers associated to the active CBF constraints.
Notice that \eqref{eq:practical_equilibrium_condition} is equivalent to $f_{\mathcal{A}}(x,\lambda_a) = 0$ as defined in \eqref{eq:fa}.
Thus, in this case, the equilibrium point is at the boundary intersection $\partial \mathcal{C}_{\mathcal{A}}$ and satisfies 
$f_{\mathcal{A}}(x_e, \lambda_a) = 0$ for some $\lambda_a \in \mathbb{R}^r_{\ge 0}$, demonstrating \eqref{eq:boundary_equilibria}.
\\
%
%------------------------------------------------------------------------------------------------
{\bf Case 4.} Consider the region where all CBF constraints are {\it inactive}:
$L_{f_{cl}} h_i + \alpha_i(h_i) > 0$, $i = 1, \cdots, N$.
Following the same steps of {\bf Case 4} of Theorem 2 on \cite{Reis2025}, we conclude that equilibrium points 
occurring in this region must lie in the interior of the safe set, that is, $x_e \in int(\mathcal{C})$.
%
% From the complementary slackness conditions, 
% $\lambda_1 \!=\! \cdots \!=\! \lambda_N \!=\! 0$. At an equilibrium point $x_e \!\in\! \mathcal{E}$, $L_{f_{cl}(x_e)} h_i(x_e) = 0$, 
% implying that $h_i(x_e) > 0$.
%
% Therefore, equilibrium points occurring in this region must lie in the interior of the safe set, that is, $x_e \in int(\mathcal{C})$.
%
Additionally, \eqref{eq:explicit_equilibrium_condition} must be satisfied with $\lambda = 0$, which means that 
$f(x_e) = p \gamma(V(x_e)) G(x_e) \nabla V(x_e)$. This demonstrates \eqref{eq:interior_equilibria}.
%
% Since all possible cases were considered, the complete set of equilibrium points is described by the union $\mathcal{E} = \cup^N_{i=1} \mathcal{E}_{\partial \mathcal{C}_i} \,\cup\, \mathcal{E}_{int(\mathcal{C})}$, thus proving \eqref{eq:equilibrium_points_set}.
%
\end{proof}

% \tref{theorem:existence_equilibria} generalizes the results in \cite{Reis_LCSS} by removing the assumption that the 
% CBFs boundaries are disjoint. It also formally demonstrates that equilibrium points exist both with the safety-filter 
% and with the CLF-CBF QP.

The work \cite{TanDimos2024} has proposed a system transformation that removes certain types of interior equilibrium points from the closed-loop system with the CLF-CBF QP controller. We conjecture that a similar transformation could be performed in the case of the safety-filter. 
Thereby, in the remaining of the paper, we focus on the stability properties for boundary equilibrium points.
\subsection{Stability of Boundary Equilibrium Points}

%------------------------------------------------------------------------------------------------
\begin{lemma}[Closed-Loop Jacobian]
\label{lemma:jacobian_boundary}
Let \eqref{eq:closed_loop} be the closed-loop system formed by the nonlinear system \eqref{eq:affine_nonlinear} 
with controller \eqref{eq:generalized_QP}, and let $\mathcal{A}$ from \dref{def:equilibrium_manifold} represent 
the indexes of $1 \le r \le n$ active CBF constraints.
For a boundary equilibrium point $x_e \in \mathcal{E}_{\partial \mathcal{C}_{\mathcal{A}}}$ with full rank 
% $U_{\mathcal{A}}(x_e)$ 
% (that is, $\nabla h_{a_1}(x_e), \cdots, \nabla h_{a_r}(x_e)$ are linearly independent) 
and $U_{\mathcal{A}}^\mathsf{T}(x_e) g(x_e) \ne 0$,
the Jacobian matrix $J_{cl}(x_e) \in \mathbb{R}^{n \times n}$ of the closed-loop system \eqref{eq:closed_loop} 
computed at $x_e$ is given by
\begin{align}
J_{f_{cl}}(x_e) 
% &\!=\! \left( I_n \!-\! G \, U A^{-1} U^\mathsf{T} \right) J_{f_\mathcal{A}}(x_e, \lambda_a)
% \!-\! G U A^{-1} D U^\mathsf{T} \label{eq:jacobian_boundary} \\
\!&=\! \mathcal{P}_{U_{\mathcal{A}}} 
\!\left( \mathcal{P}_V J_{\mathcal{A}}(x_e, \bar{\lambda}_e) \!-\! c^{-1} \gamma^{\prime}(V) G \nabla V \nabla V^\mathsf{T} \right) \nonumber
\\
&\qquad \qquad \qquad -\mathcal{P}_V G U_{\mathcal{A}} S_A^{-1} \Lambda^{\prime} U_{\mathcal{A}}^\mathsf{T}
\label{eq:jacobian_boundary} 
% \\
% \mathcal{P}_{U_{\mathcal{A}}} &= I_n \!-\! \mathcal{P}_V G U_{\mathcal{A}} \, S_A^{-1} U^\mathsf{T}_{\mathcal{A}}
% \label{eq:gen+projection} \\
% S_A &= U^\mathsf{T}_{\mathcal{A}} \mathcal{P}_V G U_{\mathcal{A}}
% \label{eq:schur_complement}
\end{align}
%
% where $\projection{U} = I_n - G \, U A^{-1} U^\mathsf{T}$ is a generalized oblique projection into the subspace generated by $\{z_1, z_2\}$. 
%
where $\bar{\lambda}_e = \mat{ p \gamma(V(x_e)) & \lambda_e^\mathsf{T} }^\mathsf{T}$, and 
$\lambda_e \in \mathbb{R}^{r}_{\ge 0}$ is the vector with the $r$ corresponding KKT multipliers at $x_e$ for the $r$ 
active CBF constraints, $\mathcal{P}_{U_{\mathcal{A}}} \!=\! I_n \!-\! \mathcal{P}_V G U_{\mathcal{A}} \, S_A^{-1} U^\mathsf{T}_{\mathcal{A}}$ with
$S_A \!\!=\! U^\mathsf{T}_{\mathcal{A}} \mathcal{P}_V G U_{\mathcal{A}}$,
$\Lambda^{\prime} \!\!=\! \Diag{ \alpha^{\prime}_{a_1}\!(h_{a_1}),\!\cdots\!,\alpha^{\prime}_{a_r}\!(h_{a_r}) } \!>\! 0$ and
\begin{align}
J_{\mathcal{A}}(x, \bar{\lambda}_a) \!&=\! \frac{\partial f_{nom}}{\partial x} \!-\! \lambda_0 \frac{\partial (G \nabla V)}{\partial x} \!+\! 
\sum_{i \in \mathcal{A}}\frac{\partial (G \nabla h_{i})}{\partial x} \lambda_{i}
\label{eq:Ja}
% D &= \begin{bmatrix}
% \gamma^{\prime}(V) & 0 & \cdots & 0\\
% 0                  & \alpha^{\prime}_{a_1}(h_{a_1}) & \cdots & 0 \\
% \vdots             & \vdots    & \ddots & \vdots \\
% 0                  & 0 & \cdots & \alpha^{\prime}_{a_r}(h_{a_r})
% \end{bmatrix} \\
% \Lambda^{\prime} &= \Diag{ \alpha^{\prime}_{a_1}(h_{a_1}) ,\, \cdots ,\, \alpha^{\prime}_{a_r}(h_{a_r}) } \nonumber
\end{align}
with $\bar{\lambda}_a = \mat{ \lambda_0 & \lambda_a^\mathsf{T} }^\mathsf{T}$. 
\end{lemma}
%------------------------------------------------------------------------------------------------
\begin{remark}
    The Jacobian $J_{\mathcal{A}}$ as defined in \eqref{eq:Ja} is not the same as $J_{f_{\mathcal{A}}}$, 
    the Jacobian of $f_\mathcal{A}$ in \eqref{eq:fa}: they differ precisely by a factor of $p \gamma^{\prime}(V) \nabla V \nabla V^\mathsf{T}$.
\end{remark}
%------------------------------------------------------------------------------------------------
\begin{proof}
This demonstration is a direct continuation of {\bf Case 3} from the proof of Theorem \ref{theorem:existence_equilibria}.
Using the complementary slackness conditions, $\dot{V} + \inner{L_g V}{u^\star}{} = \delta^\star$,
$\dot{h}_i + \inner{L_g h_i}{u^\star}{} = 0 \,\, \forall i \in \mathcal{A}$.
Substituting \eqref{eq:control_solution}-\eqref{eq:delta_solution} these expressions and using the fact that 
$\lambda_i = 0$ for all $i \notin \mathcal{A}$ yields the following system: 
\begin{align}
\underbrace{
\begin{bmatrix}
c \!&\! -\nabla V^\mathsf{T} G U_{\mathcal{A}} \\
- U_{\mathcal{A}}^\mathsf{T} G \nabla V \!&\! U_{\mathcal{A}}^\mathsf{T} G U_{\mathcal{A}}
\end{bmatrix}
}_{A_a(x)}
\underbrace{\begin{bmatrix}
\lambda_0 \\
\lambda_a
\end{bmatrix}}_{\bar{\lambda}_a}
\!=\!
\underbrace{\begin{bmatrix}
F_V \\
- F_{h_a}
\end{bmatrix}}_{b_a(x)}
\label{eq:lambda_solutions}
\end{align}
where $\lambda_a \in \mathbb{R}^r_{\ge 0}$, $F_{h_a} \!= U_{\mathcal{A}}^\mathsf{T} \, f_{nom} + \bar{\alpha}_a \in \mathbb{R}^r$ 
and $\bar{\alpha}_a \!=\! \mat{\alpha_{a_1}(h_{a_1}) \!\!&\!\! \cdots \!\!&\!\! \alpha_{a_r}(h_{a_r})}^\mathsf{T}$.
Here, $A_a(x) \in \mathbb{R}^{(r+1) \times (r+1)}$ and $b_a(x) \in \mathbb{R}^{r+1}$ are essentially reduced versions 
of matrices $A(x)$ and $b(x)$ from \eqref{eq:dualQP}, with fewer rows and columns (since not all constraints are active).
In particular, the set $\mathcal{S}_{\mathcal{A}}$ where the CLF and only the CBF constraints from $\mathcal{A}$ are active 
is given by
\begin{align}
\mathcal{S}_{\mathcal{A}} 
= \{ x \in \mathbb{R}^n \,|\, & f_{cl}(x) \!=\! f_{nom} \!-\! \lambda_0 G \nabla V \!+\! G U_{\mathcal{A}} \lambda_a \}
\label{eq:Si_set}
\end{align}
where $\lambda_0, \lambda_a$ are the positive solutions of \eqref{eq:lambda_solutions}.
%
% Solving for $\bar{\lambda}_i \in \mathbb{R}^2$ yields an expression for the two KKT multipliers.
%
Using the known formula for inversion of block matrices, since $c > 0$, $A_a(x)$ is invertible if 
its Schur complement $S_{A} = U^\mathsf{T}_{\mathcal{A}} \mathcal{P}_V G U_{\mathcal{A}} \in \mathbb{R}^{r \times r}$ 
is invertible. Since $\mathcal{P}_V(x) > 0 \, \forall x$, $S_{A}^{-1}$ exists if 
$U_{\mathcal{A}}$ has full column rank and if $U^\mathsf{T}_{\mathcal{A}}(x) g(x) \ne 0$.
Under these assumptions, a formula for $A_a^{-1}(x)$ is
\begin{align}
A_a^{-1}(x) \!=\!\! 
\begin{bmatrix}
c^{-1} \!\!&\!\! 0 \\
0 \!\!&\!\! 0
\end{bmatrix} \!+\!
\frac{1}{c^2}\!\!
\begin{bmatrix}
\nabla V^\mathsf{T} G U_{\mathcal{A}} \\
c \, I_r
\end{bmatrix}
\!S_A^{-1} \!\!
\begin{bmatrix}
\nabla V^\mathsf{T} G U_{\mathcal{A}} \\
c \, I_r
\end{bmatrix}^\mathsf{T} \!\!\!
\label{eq:Ainv}
\end{align}
where the dimensions of vectors and matrices are conformable for matrix addition and multiplication.
Then, the KKT multipliers $\lambda_0$, $\lambda_a$ can be found by solving \eqref{eq:lambda_solutions}.

Taking the derivative of \eqref{eq:lambda_solutions} with respect to the $k$-th state component $x_k$ and 
solving for $\partial_k \bar{\lambda}_a$ yields
%
% \begin{align}
% \partial_k C(x) \bar{\lambda}_i(x) \!+\! C(x) \partial_k \bar{\lambda}_i(x) = \partial_k b(x)
% \label{eq:linear_system_derivative}
% \end{align}
% %
% where the operator $\partial_k$ denotes the partial derivative with respect to $x_k$. In the case where 
% $L_g h(x) \ne 0$, $|C(x)| \neq 0$ and therefore the inverse \eqref{eq:inverseA} can be used to directly 
% solve \eqref{eq:linear_system_derivative} for the partial derivatives of $\partial_k \bar{\lambda}_i(x)$:
%
\begin{align}
\partial_k \bar{\lambda}_a(x) &= A_a^{-1} \left( \partial_k b_a - \partial_k A_a \bar{\lambda}_a \right)
\label{eq:partial_lambda} 
\end{align}
%
% To find an expression for $\partial_k \bar{\lambda}_a(x)$ using \eqref{eq:partial_lambda}, expressions for 
% $\partial_k A_a(x)$ and $\partial_k b_a(x)$ must be derived. 
%
Defining $\bar{U} = \mat{ -\nabla V \!\!&\!\! U_{\mathcal{A}}} \in \mathbb{R}^{n \times (r+1)}$,
the partial derivatives of $A_a(x)$ and $b_a(x)$ are:
\begin{align}
\partial_k A_a(x) \!&=\! \partial_k ( \bar{U}^\mathsf{T} G \bar{U} ) \label{eq:del_A} \\
\partial_k b_a(x) \!&=\! - \partial_k( \bar{U}^\mathsf{T} f_{nom} )
% &\qquad 
% \begin{bmatrix}
% \gamma^{\prime}(V) \!\!&\!\! 0 \\
% 0 \!\!&\!\! \Lambda^{\prime}
% \end{bmatrix} 
\!-\! \bar{\Lambda}^{\prime}
[\bar{U}^\mathsf{T}]_k
% \underbrace{
% \begin{bmatrix}
% - \partial_k V \\
% [U_\mathcal{A}^\mathsf{T}]_k
% \end{bmatrix} 
% }_{[\bar{U}^\mathsf{T}]_k}
\label{eq:del_b}
\end{align}
where $\bar{\Lambda}^{\prime} = \Diag{ \gamma^{\prime}(V), \alpha^{\prime}_{a_1}(h_{a_1}) ,\, \cdots ,\, \alpha^{\prime}_{a_r}(h_{a_r}) }$.
% where $[\bar{U}^\mathsf{T}]_k$ denotes the $k$-th column of $\bar{U}^\mathsf{T}$.
From \eqref{eq:Ja}, notice that $[J_\mathcal{A}]_k = \partial_k f_{nom} + \partial_k (G \bar{U}) \bar{\lambda}_a$.
Using this fact, combining equations \eqref{eq:del_A}-\eqref{eq:del_b} to compute the term $\partial_k b_a - \partial_k A_a \bar{\lambda}_a$ 
in \eqref{eq:partial_lambda}, left multiplying it by $A_a^{-1}$ and using the fact that $f_{cl}(x) = f_{nom} + G \bar{U} \bar{\lambda}_a$
yields 
% the following expression for $\partial_k \bar{\lambda}_a(x)$:
% the derivative of the KKT multipliers:
%
\begin{align}
\partial_k \bar{\lambda}_a(x) \!&=\! 
- A_a^{-1} \!\left( \bar{U}^\mathsf{T} [J_{\mathcal{A}}]_k \!+\! 
% \begin{bmatrix}
% \gamma^{\prime}(V) \!\!&\!\! 0 \\
% 0 \!\!&\!\! \Lambda^{\prime}
% \end{bmatrix} 
\bar{\Lambda}^{\prime} [\bar{U}^\mathsf{T}]_k \!+\! (\partial_k \bar{U})^\mathsf{T} f_{cl} \right) 
\label{eq:partial_lambda2} 
% \\
% &\quad - A_a^{-1} (\partial_k \bar{U})^\mathsf{T} f_{cl}(x) \nonumber
\end{align}
where the fact that $[J_\mathcal{A}]_k = \partial_k( f_{nom} + G\bar{U} \bar{\lambda}_a )$ was used.
Then, taking the derivative of the closed-loop system dynamics \eqref{eq:closed_loop} with $\lambda_i = 0 \,\, \forall i \notin \mathcal{A}$ yields
\begin{align}
\partial_k f_{cl}(x) = [J_{\mathcal{A}}]_k + G \bar{U} \partial_k \bar{\lambda}_a
\label{eq:jacobian_column}
\end{align}
which by using the expression for $\partial_k \bar{\lambda}_a$ in \eqref{eq:partial_lambda2} yields an expression for the $k$-th column 
of the closed-loop Jacobian matrix $J_{f_{cl}}(x)$ at $\mathcal{S}_{\mathcal{A}}$.
Consider an equilibrium point $x_e \in \mathcal{E}_{\partial \mathcal{C}_\mathcal{A}}$. 
Since $f_{cl}(x_e) = 0$ by definition, the last term on the right-hand side of \eqref{eq:partial_lambda} vanishes.
Then, substituting $\partial_k \bar{\lambda}_a(x_e)$ in \eqref{eq:jacobian_column} and using the fact that 
$\lambda_0(x_e) = p \gamma(V(x_e))$ yields
\begin{align}
\partial_k f_{cl}(x_e) &= (I \!-\! G \bar{U} A_a^{-1} \bar{U}^\mathsf{T} ) [J_{\mathcal{A}}(x_e, \bar{\lambda}_e)]_k 
\label{eq:jacobian_column2} \\ 
&\qquad \qquad \qquad
- G \bar{U} A_a^{-1} \bar{\Lambda}^{\prime} [\bar{U}^\mathsf{T}]_k
\nonumber
\end{align}
%
% where $\bar{\lambda}_e = \mat{ p \gamma(V(x_e)) & \lambda_e^\mathsf{T} }^\mathsf{T}$ and 
% $\lambda_e \in \mathbb{R}^{r}_{\ge 0}$ is the vector of KKT multipliers associated to 
% the $r$ active CBF constraints at $x_e$.
%
On the assumptions of the theorem, namely a full rank $U_{\mathcal{A}}(x_e)$ and $U^\mathsf{T}_{\mathcal{A}}(x_e) g(x_e) \ne 0$,
equation \eqref{eq:Ainv} can be used to simplify the following expressions at \eqref{eq:jacobian_column2}:
\begin{align}
(I - G \bar{U} A_a^{-1} \bar{U}^\mathsf{T} ) &= \mathcal{P}_{U_\mathcal{A}} \mathcal{P}_V 
\label{eq:first_term} \\
G \bar{U} A_a^{-1} 
% \begin{bmatrix}
% \gamma^{\prime}(V) \!\!&\!\! 0 \\
% 0 \!\!&\!\! \Lambda^{\prime}
% \end{bmatrix} 
\bar{\Lambda}^{\prime}
[\bar{U}^\mathsf{T}]_k &= c^{-1} \gamma^{\prime}(V) \mathcal{P}_{U_\mathcal{A}}  G \nabla V [\nabla V]_k \nonumber \\
&+ \mathcal{P}_V G U_\mathcal{A} S_A^{-1} \Lambda^{\prime} [U_\mathcal{A}^\mathsf{T}]_k
\label{eq:sec_term}
\end{align}
Substituting \eqref{eq:first_term}-\eqref{eq:sec_term} in \eqref{eq:jacobian_column2} yields
\begin{align}
\partial_k f_{cl}(x_e) \!&=\! \mathcal{P}_{U_{\mathcal{A}}} \!\left( \mathcal{P}_V J_{\mathcal{A}}(x_e, \bar{\lambda}_e) \!-\! c^{-1} \gamma^{\prime}(V) G \nabla V [\nabla V]_k \right) 
\nonumber \\
&\qquad \qquad \qquad -\mathcal{P}_V G U_{\mathcal{A}} S_A^{-1} \Lambda^{\prime} [U_{\mathcal{A}}]_k
\label{eq:jacobian_column3} 
\end{align}
which is simply the $k$-th column of the closed-loop Jacobian matrix $J_{f_{cl}}(x_e)$ computed at $x_e$.
\end{proof}
%------------------------------------------------------------------------------------------------
\begin{remark}
Notice that $U_{\mathcal{A}}^\mathsf{T} \mathcal{P}_{U_{\mathcal{A}}} = U_{\mathcal{A}}^\mathsf{T} - S_A S_A^{-1} U_{\mathcal{A}}^\mathsf{T} = 0$, therefore 
showing that $\mathcal{P}_{U_{\mathcal{A}}}$ is a generalized projection matrix for the orthogonal complement of the column space of $U_{\mathcal{A}}$.
From \eqref{eq:jacobian_boundary}, that means that $U_{\mathcal{A}}^\mathsf{T}(x_e) J_{f_{cl}}(x_e) = - \Lambda^{\prime} U_{\mathcal{A}}^\mathsf{T}(x_e)$.
Therefore, the gradients $\nabla h_{a_1}(x_e), \cdots, \nabla h_{a_r}(x_e)$ are left eigenvectors of $J_{f_{cl}}(x_e)$ with corresponding negative eigenvalues 
$-\alpha_{a_1}^{\prime}(h_{a_1}), \cdots, -\alpha_{a_r}^{\prime}(h_{a_r}) < 0$.
In particular, this implies that if an equilibrium point $x_e$ on the conditions of \tref{theorem:existence_equilibria}
occurs at the intersection of exactly $r = n$ barrier boundaries, $x_e$ must be stable.
\end{remark}
%------------------------------------------------------------------------------------------------
\begin{lemma}
\label{lemma:special_basis}
Let $X \in \mathbb{R}^{n \times n}$ be a symmetric positive definite matrix defining an inner product $\inner{\cdot}{\cdot}{X}$ over $\mathbb{R}^n$, 
and $\mathcal{Z} = \{ z_1, \cdots , z_r \}$ be a set of $r$ linearly independent vectors.
Additionally, let $\mathcal{V} = \{ v_1, \cdots, v_{n-r} \}$ be a basis for the orthogonal complement of 
$\mathcal{Z}$ with respect to $\inner{\cdot}{\cdot}{X}$. Then, $\mathcal{Z} \cup \mathcal{V}$ is a basis 
for $\mathbb{R}^n$.
\end{lemma}
%------------------------------------------------------------------------------------------------
\begin{proof}
Recall that since the span of $\mathcal{V}$ is the orthogonal complement of $\mathcal{Z}$ with respect to $\inner{\cdot}{\cdot}{X}$,
$\inner{v_i}{z_j}{X} = 0 \,\, \forall i,j$. 
Let $\sum^{r}_{i=1} a_{i} z_i \!+\! \sum^{n-r}_{i=1} b_{i} v_i \!=\! 0$ be the equation for deciding linear independence.
Taking the inner product $\inner{\cdot}{\cdot}{X}$ with $z_j \in \mathcal{Z}$ yields
% $\sum^{r}_{i=1} a_{i} \inner{z_j}{z_i}{X} \!+\! \sum^{n-r}_{i=1} b_{i} \inner{z_j}{v_i}{X} \!=\! $
$\sum^{r}_{i=1} a_{i} \inner{z_j}{z_i}{X} \!=\! 0$. Since the vectors from $\mathcal{Z}$ are linearly independent, 
the only solution is $a_i = 0,\, i = 1, \cdots, r$.
Similarly, taking the inner product $\inner{\cdot}{\cdot}{X}$ with $v_j \in \mathcal{V}$ yields
% $\sum^{r}_{i=1} a_{i} \inner{v_j}{z_i}{X} \!+\! \sum^{n-r}_{i=1} b_{i} \inner{v_j}{v_i}{X} \!=\! $
$\sum^{n-r}_{i=1} b_{i} \inner{v_j}{v_i}{X} = 0$. Since the vectors from $\mathcal{V}$ are also linearly independent 
($\mathcal{V}$ is a basis), the only solution is $b_i = 0,\, i = 1, \cdots, n-r$.
%
% Combining both sets of equations in matrix form yields 
% $
% \begin{bmatrix}
% Z^\mathsf{T} X Z \!&\! 0 \\
% 0 \!&\! V^\mathsf{T} X V
% \end{bmatrix} 
% %
% \begin{bmatrix}
% a \\ b 
% \end{bmatrix} \!=\!
% %
% \begin{bmatrix}
% 0 \\ 0 
% \end{bmatrix},
% $
% where $Z$, $V$ are matrices with the vectors from $\mathcal{Z}$, $\mathcal{V}$ as columns and $a \in \mathbb{R}^r$, $b \in \mathbb{R}^{n-r}$ 
% stack the scalars of the linear combination. 
%
Therefore, the set $\mathcal{Z} \cup \mathcal{V}$ must be composed of $n$ 
linearly independent vectors, constituting a basis for $\mathbb{R}^n$.
\end{proof}
%------------------------------------------------------------------------------------------------
\begin{proposition}
\label{prop:jac_simplified}
Let $\mathcal{W} = \{ w_1,\cdots, w_{n-1} \}$ be a basis for $\{G \nabla V\}^\perp$.
Then, with $B_V = \mat{ \nabla V \!\!&\!\! w_1 \!\!&\!\! \cdots \!\!&\!\! w_{n-1} } \in \mathbb{R}^{n \times n}$, the following formula holds:
\begin{align}
\mathcal{P}_V J_{\mathcal{A}} \!-\! c^{-1} \gamma^{\prime}(V) G \nabla V \nabla V^\mathsf{T} &= 
(B_V^\mathsf{T})^{-1} D B_V^\mathsf{T} J_{f_\mathcal{A}}
\label{eq:jacobian_formula} 
\end{align}
where $D = \Diag{ p^{-1} c^{-1}, I_{n-1}} > 0$. 
\end{proposition}
%------------------------------------------------------------------------------------------------
\begin{proof}
By \lref{lemma:special_basis}, the set $\mathcal{B}_V = \{ \nabla V, w_1, \cdots, w_{n-1} \}$ is a basis for $\mathbb{R}^n$. 
Therefore, the square matrix $B_V = \mat{ \nabla V \!\!&\!\! w_1 \!\!&\!\! \cdots \!\!&\!\! w_{n-1} }$ is full rank.
Left-multiplying matrix 
$\mathcal{P}_V J_{\mathcal{A}} \!-\! c^{-1} \gamma^{\prime}(V) G \nabla V \nabla V^\mathsf{T}$ by 
$\nabla V^\mathsf{T}$ and $w_i^\mathsf{T}$ and carrying out the algebraic simplifications due to $1 - c^{-1} \norm{\nabla V}{G}^2 = (pc)^{-1}$ and $w_i^\mathsf{T} G \nabla V = 0$ 
yields $p^{-1} c^{-1}\nabla V^\mathsf{T} J_{f_\mathcal{A}}$ and $w_i^\mathsf{T} J_{f_\mathcal{A}}$, respectively, $i = 1, \cdots, n-1$.
%
% \begin{align}
% \nabla V^\mathsf{T} \!\left( \mathcal{P}_V J_{\mathcal{A}} \!-\! \frac{1}{c} \gamma^{\prime}(V) G \nabla V \nabla V^\mathsf{T} \right) \!&=\! \frac{1}{p c} \nabla V^\mathsf{T} J_{f_\mathcal{A}} 
% \label{eq:pre_V} \\
% w_i^\mathsf{T} \!\left( \mathcal{P}_V J_{\mathcal{A}} \!-\! \frac{1}{c} \gamma^{\prime}(V) G \nabla V \nabla V^\mathsf{T} \right) \!&=\! w_i^\mathsf{T} J_{f_\mathcal{A}}
% \label{eq:pre_w}
% \end{align}
%
Combining these $n$ equations 
% \eqref{eq:pre_V}-\eqref{eq:pre_w} 
in matrix form yields \eqref{eq:jacobian_formula} with both sides left-multiplied by $B_V^\mathsf{T}$.
\end{proof}

%------------------------------------------------------------------------------------------------
Next, we demonstrate the main result of this work, a theorem for the stability of boundary equilibrium points at 
$\mathcal{E}_{\partial \mathcal{C}_\mathcal{A}}$, occurring at the intersection of exactly $r < n$ barrier boundaries.

%------------------------------------------------------------------------------------------------
\begin{theorem}[Stability of Equilibrium Points]
\label{thm:curvature}
Consider the same assumptions of \lref{lemma:jacobian_boundary}, let 
$x_e \in \mathcal{E}_{\partial \mathcal{C}_{\mathcal{A}}}$ be a boundary equilibrium point of the closed-loop system, and 
$r = |\mathcal{A}| < n$.
%
% Let the set $\mathcal{A}$ be the collection of CBF indexes corresponding to $r < n$ active CBF constraints, 
% and let $x_e \in \mathcal{E}_{\partial \mathcal{C}_{\mathcal{A}}}$ be a boundary equilibrium point with full rank 
% $U_{\mathcal{A}}(x_e)$ (that is, $\nabla h_{a_1}(x_e), \cdots, \nabla h_{a_r}(x_e)$ are linearly independent) and
% $U_{\mathcal{A}}^\mathsf{T}(x_e) g(x_e) \ne 0$.
%
If there exists $v \in \{ \nabla h_{a_1}(x_e), \cdots, \nabla h_{a_r}(x_e)\}^\perp$ 
(with the standard inner product $\inner{\cdot}{\cdot}{}$) such that
\begin{align}
v^\mathsf{T} J_{f_{\mathcal{A}}}(x_e, \lambda_e) v > 0
\label{eq:curvature_condition}
\end{align}
then $x_e$ is unstable. Otherwise, it is stable. In \eqref{eq:curvature_condition}, 
$J_{f_{\mathcal{A}}}$ is the Jacobian of the vector field \eqref{eq:fa} with respect to $x$.
\end{theorem}
%------------------------------------------------------------------------------------------------
\begin{proof}
Consider a boundary equilibrium point $x_e \in \mathcal{E}_{\partial \mathcal{C}_\mathcal{A}}$ with $U_\mathcal{A}(x_e) \ne 0$.
The first order Taylor series approximation of the closed-loop system on a neighborhood of $x_e$ is 
$\dot{x} \approx J_{cl}(x_e) \Delta x$ with $\Delta x = (x - x_e)$ being a disturbance vector around the equilibrium point. 
Since the CBF gradients are linearly independent at $x_e$ by assumption, $\Delta x$ can be written 
using a basis $\{ \nabla h_{a_1}(x_e), \cdots, \nabla h_{a_r}(x_e), v_1, \cdots, v_{n-1}(x_e) \}$, 
where $v_1, \cdots, v_{n-r}$ are fixed basis vectors for $\{ \nabla h_{a_1}(x_e), \cdots, \nabla h_{a_r}(x_e) \}^\perp$. 
Therefore, $v_j^\mathsf{T} U_\mathcal{A}(x_e) = 0 \,\, \forall j$ by construction.
%
% Collecting these vectors as columns in a matrix $V \in \mathbb{R}^{n \times (n-r)}$, 
%
One can write $x = x_e + U_\mathcal{A}(x_e) \, a + v$, where $v$ is a linear combination of the 
$v_i$, $i = 1, \cdots, n-r$ and $a \in \mathbb{R}^r$ is a vector of coordinates.
The time derivative of $x$ in this new basis is $\dot{x} = U_\mathcal{A}(x_e) \, \dot{a} + \dot{v}$. 
Left-multiplying this equation and $J_{f_{cl}}(x_e) \Delta x$ by $U_\mathcal{A}^\mathsf{T}(x_e)$, and 
using the expression \eqref{eq:jacobian_boundary} for the closed-loop Jacobian at $x_e$ and the fact 
that $U_\mathcal{A}^\mathsf{T} v = 0$ yields
\begin{align}
U_\mathcal{A}^\mathsf{T} \dot{x} &= U_\mathcal{A}^\mathsf{T} (U_\mathcal{A} \, \dot{a} + \dot{v} ) = 
U_\mathcal{A}^\mathsf{T} U_\mathcal{A} \, \dot{a} 
\label{eq:first_expr} \\
U_\mathcal{A}^\mathsf{T} J_{f_{cl}}(x_e) \Delta x 
&= - \Lambda^{\prime} U_\mathcal{A}^\mathsf{T} (U_\mathcal{A} a + v) = - \Lambda^{\prime} U_\mathcal{A}^\mathsf{T} U_\mathcal{A} a
\label{eq:second_expr}
\end{align}
Comparing \eqref{eq:first_expr}-\eqref{eq:second_expr} yields 
$\dot{a} = - (U_\mathcal{A}^\mathsf{T} U_\mathcal{A})^{-1} \Lambda^{\prime} U_\mathcal{A}^\mathsf{T} U_\mathcal{A} \, a \in \mathbb{R}^r$.
Since $\Lambda^{\prime} > 0$, this subsystem is asymptotically stable, which means that the column space of $U_\mathcal{A}$ 
is contained in the stable subspace associated to $x_e$.
Using these equations to solve for the dynamics of $v$ and letting $a \rightarrow 0$, one can conclude 
that the stability of $x_e$ is fully determined by the subsystem $\dot{v} = J_{f_{cl}}(x_e) v$.
% where $v \in \{ \nabla h_{a_1}(x_e), \cdots, \nabla h_{a_r}(x_e) \}^\perp$.

The corresponding Lyapunov equation for $J_{cl}(x_e)$ is
%
%Computing its time derivative yields
%$\dot{V} = \dot{x}^\mathsf{T} X x + x^\mathsf{T} X \dot{x} = x^\mathsf{T} Y x$, where
%
\begin{align}
Y &= J_{cl}(x_e)^\mathsf{T} X + X J_{cl}(x_e)
\label{eq:lyap_equation} 
\end{align}
%
% where $J_{cl}(x_e)$ is expressed by \eqref{eq:jacobian_boundary}.
%
with $X = U_\mathcal{A} \Lambda_a U_\mathcal{A}^\mathsf{T} + W \Lambda_w W^\mathsf{T}$, 
where $\Lambda_a , \Lambda_w > 0$ are diagonal matrices and the column space of $W$ is the 
orthogonal complement of $\{ \nabla h_{a_1}(x_e), \cdots, \nabla h_{a_r}(x_e) \}$ with an inner product induced by 
$\mathcal{P}_V G > 0$, that is, $U_\mathcal{A}^\mathsf{T} \mathcal{P}_V G W = 0$. This means that $X > 0$.
%
% \begin{align}
% X &= Z \Lambda_z Z^\mathsf{T} + W \Lambda_w W^\mathsf{T} > 0
% \label{eq:X}
% \end{align}
%
Using \eqref{eq:jacobian_boundary} and \pref{prop:jac_simplified}, notice that 
$X J_{f_{cl}}(x_e) v = W \Lambda_w W^\mathsf{T} (B_V^\mathsf{T})^{-1} D B_V^\mathsf{T} J_{f_\mathcal{A}} v$, where 
again $v$ is an arbitrary vector in $\{ \nabla h_{a_1}(x_e), \cdots, \nabla h_{a_r}(x_e)\}^\perp$.
Define the Lyapunov candidate $V(v) = v^\mathsf{T} X v > 0$. Taking its time derivative and using the dynamics of $v$ yields
\begin{align}
\dot{V} &= v^\mathsf{T} \left( J_{f_{cl}}^\mathsf{T} X + X J_{f_{cl}} \right) v \nonumber \\
&= 2 v^\mathsf{T} W \Lambda_w W^\mathsf{T} (B_V^\mathsf{T})^{-1} D B_V^\mathsf{T} J_{f_\mathcal{A}} v
\label{eq:dotV_stability} 
\end{align}
Since $(B_V^\mathsf{T})^{-1} D B_V^\mathsf{T}$ is similar to $D > 0$, its eigenvalues are $p^{-1}c^{-1} > 0$ and ones.
Since $\{ \nabla h_{a_1}(x_e), \cdots, \nabla h_{a_r}(x_e)\}^\perp$ is contained in the column space of $W \Lambda_w W^\mathsf{T}$,
given any $v \in \{ \nabla h_{a_1}(x_e), \cdots, \nabla h_{a_r}(x_e)\}^\perp$, it is always possible to choose 
$W \Lambda_w W^\mathsf{T} \ge 0$ such that $v$ is one of its eigenvectors with an associated positive eigenvalue 
$\sigma > 0$.
Therefore, with this choice for $\Lambda_w > 0$ and $W$, \eqref{eq:dotV_stability} becomes 
$\dot{V} = 2 \sigma v^\mathsf{T} J_{f_\mathcal{A}} v$. 
Hence, by Chetaev's instability theorem, if there exists $v \in \{ \nabla h_{a_1}(x_e), \cdots, \nabla h_{a_r}(x_e)\}^\perp$ 
such that $v^\mathsf{T} J_{f_\mathcal{A}} v > 0$ holds, $\dot{V} > 0$ and $x_e$ is an unstable equilibrium point.
Otherwise, $\dot{V} \le 0$ and $x_e$ is stable. 
% In particular, if $v^\mathsf{T} J_{f_\mathcal{A}} v < 0$ holds with an strict inequality for all 
% $v \in \{ \nabla h_{a_1}(x_e), \cdots, \nabla h_{a_r}(x_e)\}^\perp$, $x_e$ is asymptotically stable.
\end{proof}
%------------------------------------------------------------------------------------------------

Theorems \ref{theorem:existence_equilibria} and \ref{thm:curvature} show that the existence conditions 
and stability properties of boundary equilibrium points are completely determined by the vector field 
$f_{\mathcal{A}}$ and its state derivatives as defined in \eqref{eq:fa}.
For both frameworks for safety-critical control considered in this work, the stability of boundary 
equilibrium points depends on the state derivatives of the system dynamics $f$ and $g$ 
and on the Hessians of the active CBFs, $H_{h_{a_1}}, \cdots, H_{h_{a_r}}$. Particularly, for the 
safety filter QP, it also depends on the Jacobian of the nominal controller $J_{u_{nom}}(x)$, and 
for the CLF-CBF QP, it also depends on the Hessian matrix of the CLF, $H_V$.

\begin{figure}[htbp]
\centering
\vspace{-4mm}
\includegraphics[width=0.90\columnwidth]{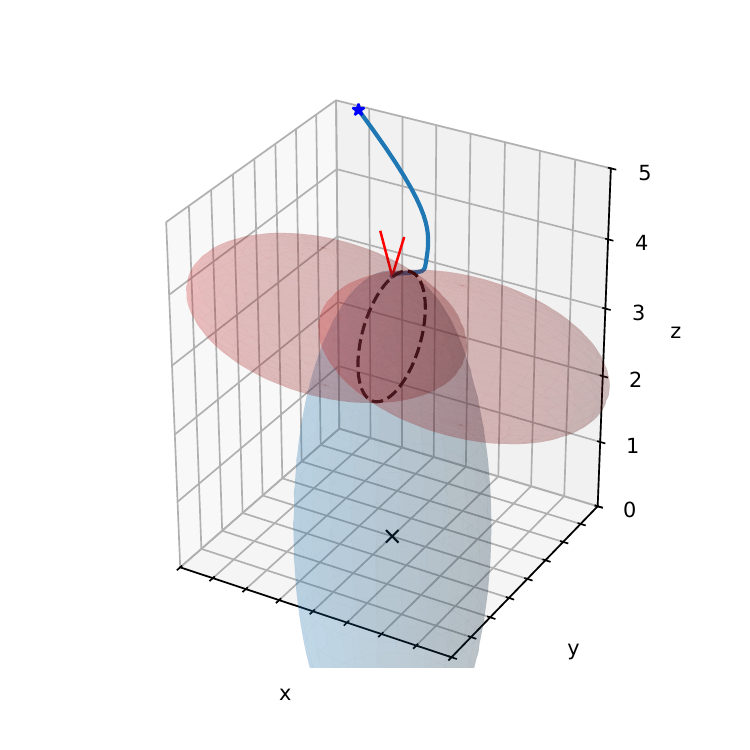}
\vspace{-8mm}
\caption{Example of a asymptotically stable equilibrium point occurring at the intersection of two quadratic CBF boundaries.}
\label{fig:equilibrium_pt}
\end{figure}

In \fref{fig:equilibrium_pt}, we show the simulation result of a safety-critical control task with the CLF-CBF QP-based controller.
\footnote{
The code repository used for producing these results is publicly available at \url{https://github.com/CaipirUltron/CompatibleCLFCBF/tree/mydevel}.}.
Here, $u_{nom} = 0$ and $H(x) = I_3$ in \eqref{eq:generalized_QP}, and the proposed system in $\mathbb{R}^3$ is 
$\dot{x}_1 = u_1 - 2 u_3$, $\dot{x}_2 = u_2$, $\dot{x}_3 = -2 u_1 + u_3$.
The CBFs are two quadratic functions $h_1$ and $h_2$ centered at the points $(-1, 0, 3)$ and $(1, 0, 3)$, respectively. 
Their boundaries are the red ellipsoids in \fref{fig:equilibrium_pt}, with the union of their interiors constituting the unsafe set which the system trajectories should avoid.
The CLF $V$ is also a quadratic centered on the origin, and its level set is shown in blue, at an asymptotically stable 
equilibrium point $x_e$ at the top of the boundary intersection $\partial \mathcal{C}_{\mathcal{A}} = \partial \mathcal{C}_1 \cap \partial \mathcal{C}_2$, 
shown in \fref{fig:equilibrium_pt} by the dashed black circle.
A trajectory converging to $x_e$ is shown, and the normalized gradients $\nabla h_1(x_e)$, $\nabla h_2(x_e)$ are shown as the two red vectors pointing up at the equilibrium point.
At $x_e$, since $f_\mathcal{A}(x_e, \lambda_e) = 0$ for some $\lambda_e = \mat{ \lambda_{e_1} \!&\! \lambda_{e_2} }^\mathsf{T} \in \mathbb{R}^2_{\ge 0}$, 
$f(x_e) = p \gamma(V(x_e)) \nabla V(x_e) - G(x_e)( \lambda_{e_1} \nabla h_1(x_e) + \lambda_{e_2} \nabla h_2(x_e) )$.
Since the system is driftless, $f(x) = 0 \,\, \forall x$. Furthermore, $G(x_e)$ is a positive definite matrix. Then, the following holds:
$p \gamma(V(x_e)) \nabla V(x_e) = \lambda_{e_1} \nabla h_1 + \lambda_{e_2} \nabla h_2$, meaning that the gradient of the CLF is a conical combination of the 
gradients of the active CBFs at $x_e$.
That is precisely the case at \fref{fig:equilibrium_pt}. Furthermore, due to the fact that the system is driftless with a full rank $G(x_e)$, 
carrying out the needed simplifications at $v^\mathsf{T} J_{f_{\mathcal{A}}}(x_e,\lambda_e) v$ from \eqref{eq:curvature_condition}, 
one can conclude that the stability of $x_e$ is determined by $\lambda_{e_1} H_{h_1} \!+\! \lambda_{e_2} H_{h_2} \!-\! p \gamma(V(x_e)) H_V$, that is, 
essentially by a difference of curvatures between the CBFs and the CLF at $x_e$, extending the result in \cite{Reis_LCSS}.
% \input{section4}

%%%%%%%%%%%%%%%%%%%%%%%%%%%%%%%%%%%%%%%%%%%%%%%%%%%%%%%%%%%%%%%%%%%%%%%%%%%%%%%%
% \section{CONCLUSIONS}

%                                   % the textheight of the last page by a suitable amount.
%                                   % This command does not take effect until the next page
%                                   % so it should come on the page before the last. Make
%                                   % sure that you do not shorten the textheight too much.

%\section*{ACKNOWLEDGMENTS}
\bibliographystyle{plain}
\bibliography{root}

\end{document}